\documentclass[11pt]{article}
\usepackage{babel}
\usepackage[utf8]{inputenc}
\usepackage[T1]{fontenc}
\usepackage[a4paper, margin=3.5cm]{geometry}
\usepackage{amsmath, amsfonts, amssymb, amsthm, subcaption, hyperref}
\usepackage{graphicx}
\usepackage{authblk}
\usepackage{lmodern}

\usepackage{lmodern}

\newtheorem{theorem}{Theorem}
\newtheorem{definition}{Definition}

\newtheorem{conjecture}{Conjecture}

\usepackage[backend=biber, style=alphabetic, sorting=nymt, maxbibnames=5]{biblatex} %% for bibliography
\DeclareSortingTemplate{nymt}{ % custom sorting for biblio
  \sort{
    \field{presort}
  }
  \sort[final]{
    \field{sortkey}
  }
  \sort{
    \field{sortname}
    \field{author}
  }
  \sort{
    \field{sortyear}
    \field{year}
  }
  \sort{
    \field[padside=left,padwidth=2,padchar=0]{month}
    \literal{00}
  }
  \sort{
    \field{sorttitle}
 }
}
\graphicspath{{./figures/}}

\newcommand{\Nant}{\boldsymbol{\uparrow}}
\newcommand{\Eant}{\boldsymbol{\rightarrow}}
\newcommand{\Want}{\boldsymbol{\leftarrow}}
\newcommand{\Sant}{\boldsymbol{\downarrow}}

\newcommand{\turnleft}{\mathrm{t^\circlearrowleft}}
\newcommand{\turnright}{\mathrm{t^\circlearrowright}}
\newcommand{\tendsto}{\underset{n\to\infty}{\longrightarrow}}

\newcommand{\vv}[1]{\overrightarrow{#1}}

\begin{document}
\title{Sideways on the highways}
%% \titlerunning{Sideways on the highways}
%% \author{Victor H. Lutfalla\inst{1}\thanks{Partial support from ANR IZES and ANR ALARICE}}
%% \authorrunning{V. Lutfalla}
%% \institute{Aix-Marseille Université}
\author[1,2]{Victor Lutfalla}
\affil[1]{Université Publique, France}
\affil[2]{Aix-Marseille Université, CNRS, I2M, Marseille, France}
\date{2025}
\maketitle

\begin{abstract}
  We present two generalised ants ($LLRRRL$ and $LLRLRLL$) which admit both highway behaviours and other kinds of emergent behaviours from initially finite configurations.
  This limits the well known Highway conjecture on Langton's ant as it shows that a generalised version of this conjecture generically does not hold on generalised ants.
\end{abstract}

%% \begin{keywords}
%%   Langton's ant, emergent behaviour, highway conjecture, turmites
%%   \end{keywords}


\paragraph{Aknowledgements.} The author aknowledges partial support from ANR IZES and ANR ALARICE.

\section{Introduction}
The original of \emph{Langton's ant} model was introduced in the 80's by different researchers \cite{langton1986, dewdney1989} as a model for ``artificial life'' meaning a simple discrete dynamical system exhibiting some kind of \emph{emergent behaviour}.

\begin{figure}[h!]
  \center
  \includegraphics[width=0.5\textwidth]{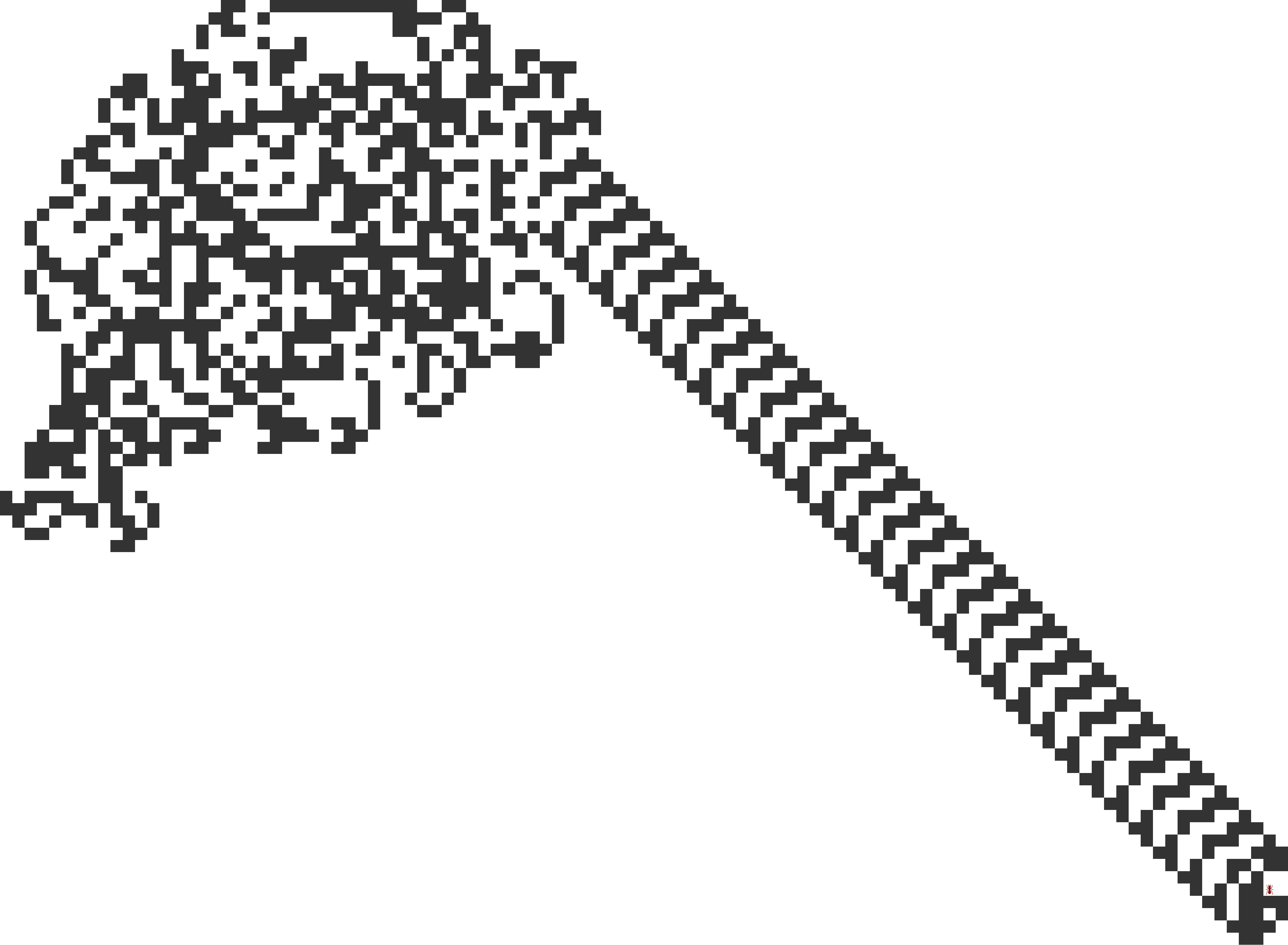}
  \caption{The configuration reached by Langton's ant from the uniform configuration after $13\,000$ steps. We can observe a periodic pattern leaving the initial seemingly chaotic pattern.}
  \label{fig:LR_highway}
\end{figure}

This emergent behaviour, called the \emph{highway}, is periodic with a drift.
From the initially uniform configuration, Langton's ant has, at first, a seemingly chaotic behaviour and after almost $10\, 000$ steps it enters the highway behaviour described above, see Fig.~\ref{fig:LR_highway}.

The remarkable observation, which dates from the first simulation in the 80s, is that every simulation of Langton's ant starting from a \emph{finite} configuration (only a finite number of cells are in a non-zero state) appears to have the same qualitative behaviour : an initial seemingly chaotic behaviour until the ant creates the suitable local pattern that starts the same highway of period 104 and shift $(\pm 2, \pm 2)$.

The \emph{highway conjecture} states that it is the case from every finite configuration.
Despite a lot of both theoretical and experimental efforts, little progress has been made towards proving or disproving this conjecture.

Several generalisations have been proposed for this model.
%% : with more than one ant~\cite{langton1986,beuret1997}, with more movements~\cite{bunimovich1992}, with more grid states~\cite{gale1994}, with more internal states and complex behaviour~\cite{dewdney1989, salo2023}, on other grids~\cite{gajardo2002}, and on other dimensions~\cite{dorbec2008}.
The generalisation that seems to better preserve the original ant's properties is the one called \emph{generalised ants}~\cite{gale1994,gale1998}.
In this class, emergent behaviours also appear and some of them also seem to be unavoidable (though no proof was found).

In the present paper we show that a relaxed generalisation of the highway conjecture does not hold over generalised ant.
In \cite{gajardo2025} we showed that there is no universal bound on the number of possible distinct highway behaviours for a given ant rule and even that there exist ant rules with infinitely many highways.
However, on the ants we discussed in that paper, the only asymptotic behaviour we prove (and observe) are diagonal highways \emph{i.e.} the behaviour is periodic with a drift which is a multiple of $(\pm 1, \pm 1)$. So one may argue that these are qualitatively similar behaviours.
In this paper we present two ants, which we found during a mass exploration of generalised ants, that admit qualitatively different asymptotic behaviours : the $LLRRRL$ ant and the $LLRLRLL$ ant.
We show that the $LLRRRL$ admits diagonal highways, non-diagonal highways and ``increasing-rectangle'' asymptotic behaviours (see Fig.~\ref{fig:pattern_rectangle}) from finite configurations; and that the $LLRLRLL$ admits an infinite family of highways and ``conic'' asymptotic behaviours (see Fig.~\ref{fig:cone}) from finite configurations.

In Section~\ref{sec:ants} we introduce the model and definitions.
In Section~\ref{sec:llrrrl} we present our results on the $LLRRRL$ ant. %% REMARK and simulations
In Section~\ref{sec:llrlrll} we present our results on the $LLRLRLL$ ant. %% REMARK and simulations 

We provide a simulator to illustrate the behaviours we describe which is hosted on \href{https://lutfalla.fr/ant}{lutfalla.fr/ant} and whose source code is publicly accessible \cite{lutfalla2025}.

\section{Settings}
\label{sec:ants}
\subsection{Non-trivial ants}

Generalised ants are a generalisation of Langton's ant defined by a \emph{rule word} $w$ over alphabet $\{L,R\}$ moving over a configuration $c\in \{0,1,\dots |w|-1\}^{\mathbb{Z}^2}$.
When the ant arrives on a cell in state $k$, it turns clockwise if $w_k=R$ and counterclockwise otherwise, then increases the cell state by $1$ modulo $|w|$ and moves one step forward, see Fig.~\ref{fig:def_glant}.

\begin{figure}[htp]
  \center
  \includegraphics[width=\textwidth]{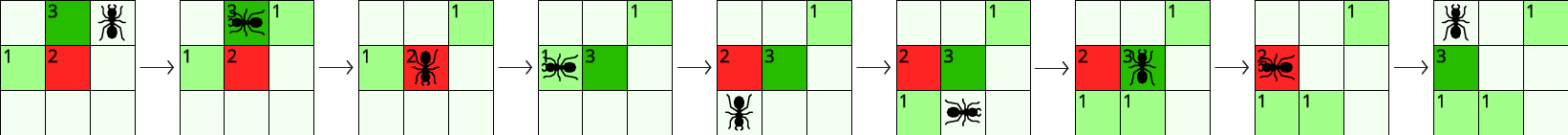}
  \caption{A simple example of the behaviour of the $LLRL$ ant. Left turning states are shown in shades of green, and the right turning state is in red. Non-zero states are additionally labelled with their state number.}
  
  \label{fig:def_glant}
\end{figure}

We now give a formal definition. \pagebreak

\begin{definition}[Internal states]
  We define the set $Q$ of directions (or internal state) $Q=\{\Nant,\Eant,\Sant,\Want\}\equiv \{(0,1), (1,0), (0,-1), (-1,0)\}$.
  We define the positive turn function $\turnleft$ on $Q$ also called counterclockwise permutation, and its converse the negative turn function $\turnright$ or clockwise permutation.
  That is \[\turnleft(\Nant)\ =\ \Want\ =\ \turnright(\Sant),\  \dots\]
\end{definition}

For simplicity we write $(i,j)+q$ with $q \in Q$ in that case we consider $Q$ through the natural isomorphism to $\{(0,1), (1,0), (0,-1), (-1,0)\}$.%% , that is : $\Eant\, \equiv (0,1)$, \dots.

\begin{definition}[Generalised ant of rule word $w$]
  Given a finite word $w\in \{L,R\}^*$.
  We call ant configuration a tuple $(c, (i,j), q)$ with $c\in \{0,1,\dots |w|-1\}^{\mathbb{Z}^2}$ called the grid configuration, $(i,j)\in \mathbb{Z}^2$ called the ant position, and $q\in Q$ the ant direction (or internal state).
  We define the generalised ant of rule word $w$ as a transition function $T_w$ over ant configurations defined as $T_w(c,(i,j),q)=(c',(i',j'),q')$ with :
  \begin{itemize}
  \item denoting $k=c_{i,j}$, if $w_k=R$ then $q':= \turnright(q)$, otherwise $q':=\turnleft(q)$
  \item $(i',j'):=(i,j)+q'$
  \item $c'_{i,j}:= c_{i,j} + 1 \mod |w|$
  \item $\forall (k,l)\neq(i,j), c'_{k,l}:=c_{k,l}$
  \end{itemize} 
\end{definition}

We call \emph{trace} of ant $w$ from configuration $C=(c,(i,j),q)$ the sequence of symbols (grid cell states) that the ant encounters when evolving from $C$, and \emph{trajectory} of $w$ from $C$ the sequence of coordinates of the ant when evolving from $C$.
The trace, which is an infinite word on alphabet $\{0,1,\dots |w|-1\}$, contains all the information about the behaviour of the ant and the initial configuration (or at least on the portion of $\mathbb{Z}^2$ that was visited by the ant).

Very few formal result are known on the model of generalised ant, for a quick overview of those results see \cite{gajardo2025}.

We now define highways and other emergent behaviours.

%% configurations manipulations ? local pattern, inclusion, etc ? 

\subsection{Emergent behaviour}

For simplicity and conciseness, we define emergent behaviours only through the trace and trajectory of ants from finite configurations.
%% FIXME formal def to appendix

The easiest emergent behaviour to describe are highways, that is, periodic behaviours with a drift.

\begin{definition}
We say that ant $w$ is in a \emph{highway} of period $n$ and drift $(a,b)$ from configuration $C=(c,(i,j),q)$ if :
\begin{itemize}
\item $c$ is finite, that is : $|\{(i,j)\in \mathbb{Z}^2, c_{i,j} \neq 0\}|< +\infty$
\item the trace $x\in \{0,\dots |w|-1\}^\mathbb{N}$ of ant $w$ from $C$ is $n$-periodic\footnote{We also assume that the trace $x$ is not $k$-periodic for $k<n$.}, that is : for any $i$ we have $x_{i+n} = x_i$
\item the trajectory $\mathbf{y} \in (\mathbb{Z}^2)^\mathbb{N}$ of ant $w$ from $C$ is $n$-periodic modulo an $(a,b)$-drift, that is : for any $i$ we have $\mathbf{y}_{i+n} = \mathbf{y}_i + (a,b)$
\end{itemize}
\end{definition}

%% , but event those have quite a verbose formal definition.

%% \begin{definition}
%%   We say that and $w$ is in a \emph{highway} of period $n$ and drift $(a,b)$ from configuration $C=(c,(i,j),q)$ if there exists a finite domain $D\subset \mathbb{Z}^2$ containing $(i,j)$ such that, denoting $(c',(i',j'),q')=T_w^n(C)$ we have :
%%   \begin{itemize}
%%   \item $(i',j')=(i,j)+(a,b)$
%%   \item $q'=q$
%%   \item denoting $(i,j)=(i_0,j_0) , (i_1,j_1),\dots (i_{n-1},j_{n-1})=(i',j')$ the succesive positions of the ant for the first $n$ steps starting from $C$ we have $\{(i_l,j_l), 0 \leq l < n\} \subseteq D$
%%   \item $c'$ restricted to $D+(a,b)$ is equal to $c$ restricted to $D$
%%   \item $c'$ restricted to the tube $(D+\{l(a,b), l\in \mathbb{N}, l\geq 1\})\setminus (D+(a,b))$ is the $0$-uniform configuration.
%%   \end{itemize}
%% \end{definition}

By analogy, we say that the ant $w$ starts a highway of period $n$ and drift $(a,b)$ from pattern $P$ if it start such a highway from the finite configuration obtained by perturbating the $0$-uniform configuration by pattern $P$ at the origin.

We consider that two highways for ant $w$ are equivalent when, their trace and trajectory are identical up to removing a finite prefix from both.
And we consider highways up to this equivalence relation.

We also say that ant $w$ eventually enters or reaches a highway from configuration $C$ if there exists an integer $N$ such that $w$ starts a highway at configuration $T_w^N(C)$.

We now introduce more general notions on traces and on trajectory in order to define our two next emergent behaviour.
\begin{definition}[Interlaced word]
  An infinite word $x\in \mathcal{A}^\mathbb{N}$ on finite alphabet $\mathcal{A}$ is called \emph{interlaced} when there exists an integer $k$, and words $u_i$, $v_i\in \mathcal{A}^*$ for $0\leq i < k$ such that $x=\bigodot\limits_{n \in \mathbb{N}} t_n = t_0\cdot t_1 \cdot t_2 \cdot \dots t_n \dots $ where $t_n= \bigodot\limits_{0\leq i < k} u_i\cdot\left(v_i^n\right) = u_0(v_0^n)u_1(v_1^n)\dots u_{k-1}(v_{k-1}^n)$.
\end{definition}

%% We say that an infinite word $x$ is \emph{eventually interlaced} if $x=u\cdot y$ with $u$ a finite (prefix) word and $y$ in linearly interlaced.

Examples of interlaced words include periodic words (case $k=1$, $v_0$ is the empty word), and words of the form $uvuvvuvvvuvvvuv^4uv^5\dots$ (case $k=1$) the simplest non periodic such word is $01011011101111…$ ($u=01$ and $v=1$).%%  With $k=2$ we then have $u_1u_2u_1v_1u_2v_2u_1v_1v_1u_2v_2v_2\dots$ .

%% In particular if the ant $w$ from configuration $C$ eventually enters a highway behaviour, then the trace $x$ of ant $w$ from $C$ is eventually periodic and therefore eventually linearly interlaced.

\begin{definition}[Previsible trajectory]
  An infinite trajectory $\mathbf{y} \in (\mathbb{Z}^2)^{\mathbb{N}}$ is \emph{$\Delta$-previsible} (or simply \emph{previsible}) for a compact $\Delta$ containing the origin $0$ when there exists two non-decreasing functions $f,g : \mathbb{N} \to \mathbb{R}^+$, an integer $k$ and a non-zero vector $\vv{d} \in \mathbb{R}^2$ such that:
  \begin{itemize}
  \item $\partial \Delta$ contains $\mathbf{y}$  up to $f$, $g$, $\vv{d}$ and  $k$, that is :
    \[ \forall n \in \mathbb{N},\  \mathbf{y}_n \in \partial_k(f(n)\Delta + g(n)\vv{d})\]
      %% where $\partial_k(X)$, called the $k$-boundary of $X$, is the set of points of $\mathbb{R}^2$ which are both at distance at most $k$ from $X$ and from its complement,  $f(n)\Delta$ is the image of $\Delta$ by the expansion of factor $f(n)$ and center $0$, and $+g(n)\vv{d}$ is the translation of vector $g(n)\vv{d}$.
    \item $\mathbf{y}$ exhausts $\partial \Delta$ up to $f$, $g$, $\vv{d}$ and $k$, that is :
      \[ \forall p \in \partial \Delta,\ \forall N \in \mathbb{N},\ \exists n\geq N,\ d\left(\mathbf{y}_n, f(n)p+ g(n)\vv{d}\right) \leq k \]
  \end{itemize}
  Where $f(n)X+g(n)\vv{d}$ is the translation by $g(n)\vv{d}$ of the expansion of $X$ by factor $f(n)$ centred on $0$, $\partial X$ is the boundary of $X$ and $\partial_k X$ is the set of points that are at distance at most $k$ from $\partial X$.

When $f(n)\tendsto +\infty$ we say that the behaviour is \emph{increasing} (and \emph{non-increasing} otherwise), and when $g(n) \tendsto +\infty$ we say that the behaviour is \emph{drifting} (and \emph{non-drifting} otherwise).
\end{definition}

With the additional assumption that functions $f$ and $g$ are easily computable (usually quadratic or quasi-linear functions), a $\Delta$-previsible trajectory becomes easily previsible in terms of computation in the sense of \cite{gajardo2017}, which means that the prediction problem (predicting the state of a given cell after a given number of steps) is computationally simple.

\begin{definition}[Emergent behaviour]
  We say that ant $w$ has an \emph{emergent behaviour} from configuration $C$ if its trace is eventually interlaced and its trajectory is $\Delta$-previsible for some compact $\Delta$.
\end{definition}

Highways are the case of drifting non-increasing emergent behaviour, and the classical result on unboundedness of ant trajectories \cite{bunimovich1992} implies that the case of non-drifting non-increasing emergent behaviour does not exist.

\begin{definition}[Increasing rectangle behaviour]
  We say that ant $w$ has an \emph{increasing rectangle behaviour} from configuration $C$ if it has an increasing non-drifting emergent behaviour where the compact $\Delta$ is a rectangle.
\end{definition}

\begin{definition}[Cone behaviour]
  We say that ant $w$ has a \emph{cone behaviour} from configuration $C$ if it has a drifting and increasing emergent behaviour with the drifting function $g$ and increasing function $f$ of the same order, that is $g=\Theta(f)$.
\end{definition}

%% We say that ant $w$ has generalised increasing rectangle behaviour (resp. generalised cone behaviour) when we drop the trace condition and only consider the trajectory conditions.

Note that our definition of emergent behaviour is quite restrictive and does not include some known behaviours that are considered emergent, such as the increasing square built by ant $LRRRRLLLRRRL$ from the $0$-uniform configuration (which has a relatively simple trace word, but not precisely interlaced).

To build simple emergent behaviour we use \emph{widgets}, that is : small patterns that have a simple ant dynamics and which when combined with other widgets yield complex (but easily decomposable) dynamics.

\begin{figure}[htp]
    \center \includegraphics[width=0.48\textwidth]{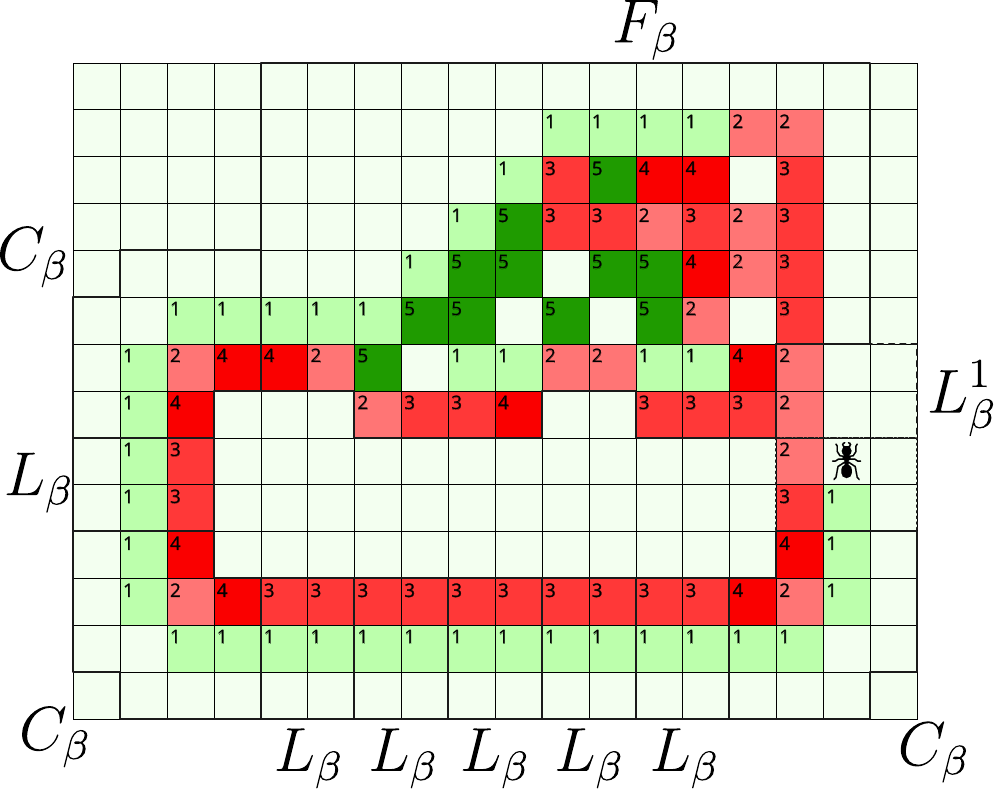}
    \hfill
    \includegraphics[width=0.48\textwidth]{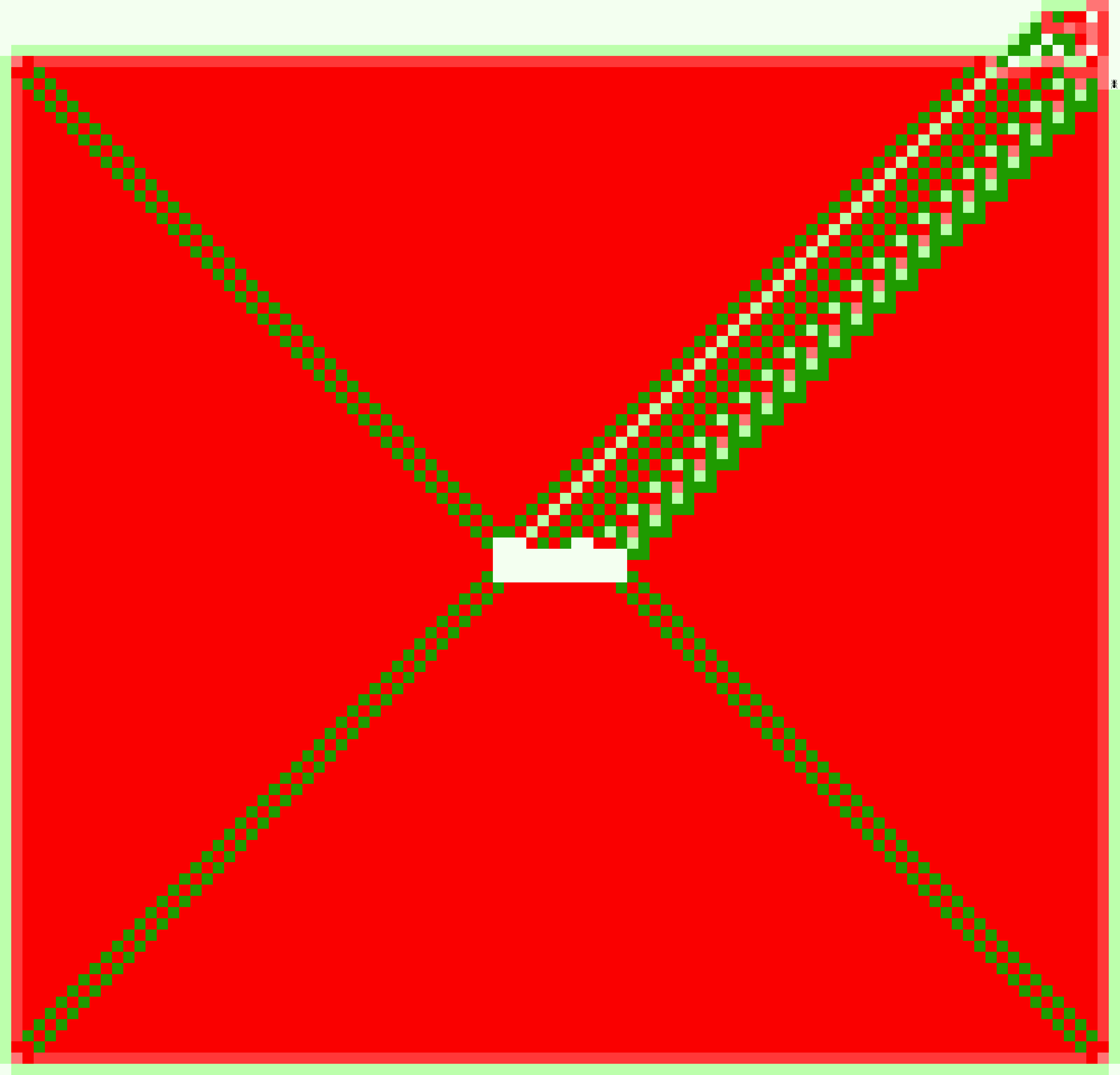}
    \caption{The pattern $P_\beta$ (left) on which the ant $LLRRRL$ starts an increasing rectangle behaviour, and the configuration reached from $P_\beta$ (right, scaled down) after $39480$ steps ($20$ full revolutions).}
    \label{fig:pattern_rectangle}
\end{figure}
\pagebreak

\section{Emergent behaviours of the $LLRRRL$ ant }
\label{sec:llrrrl}
In our previous work we presented a family of ants that has a growing set of highways and even a single ant that has an infinite set of highways \cite{gajardo2025}.
However, as many people remarked, all the highways that we discussed in that work are diagonal, that is their drift is a multiple of $(\pm 1, \pm 1)$ so one may argue that they are quite similar.

In this section we present the ant $LLRRRL$ that has multiple qualitatively different emergent behaviours.

\begin{theorem}
  The ant $LLRRRL$ has (from finite configurations) :
  \begin{enumerate}
  \item non-diagonal highway behaviours
  \item increasing rectangle behaviours
  \item an infinite family of diagonal highway behaviours
  \end{enumerate}
\end{theorem}

\begin{proof}[Sketch]
  %% As is usual with ants, the proof of items 1 and 2 of this theorem is surprisingly easy once you find the suitable seed. Item 3 needs a bit more proof.
  As is usual with ants, the proof of item 1 is surprisingly easy once you find the suitable seed pattern as one only has to compute the dynamics of the $LLRRRL$ ant from the prescribed pattern for the given number of steps and check that it has reached the same local pattern up to a drift.
  Items 2 and 3 require the use of widgets. We only describe the widgets and their dynamics in this sketched proof.%% FIXME and we give more details in the appendix.

  Starting on the pattern $P_\alpha$ of Fig.~\ref{fig:pattern_horizontal}, the ant $LLRRRL$ starts a horizontal highway\footnote{\href{https://lutfalla.fr/ant/highway.html?antword=LLRRRL&t=horizontal}{\texttt{lutfalla.fr/ant/highway.html?antword=LLRRRL\&t=horizontal}}} of period $300$ and drift $(-2,0)$.

   \begin{figure}[htp]
    \center \includegraphics[width=0.5\textwidth]{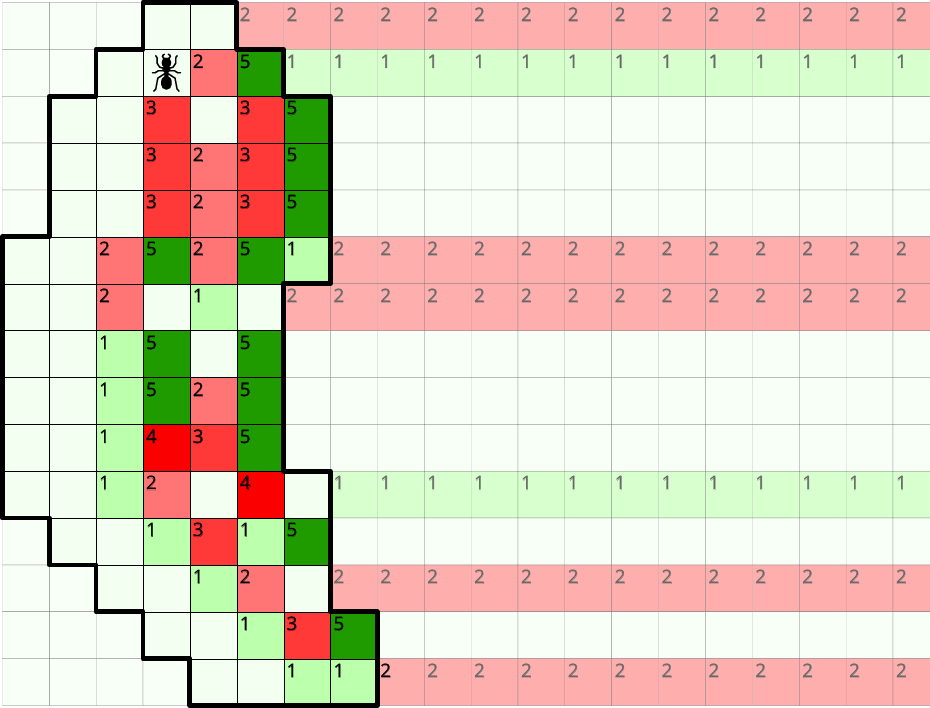}
    \caption{The pattern $P_\alpha$ (with bold boundary) on which the ant $LLRRRL$ starts a highway of period $300$ and drift $(-2,0)$. To the left of $P_\alpha$ we see the \emph{residue} left behind by the highway.}
    \label{fig:pattern_horizontal}
    \end{figure}

  Starting on the pattern $P_\beta$ of Fig.~\ref{fig:pattern_rectangle}, the ant $LLRRRL$ starts an ``increasing rectangle'' behaviour\footnote{\href{https://lutfalla.fr/ant/highway.html?antword=LLRRRL&t=rectangle}{\texttt{lutfalla.fr/ant/highway.html?antword=LLRRRL\&t=rectangle}}}.

  The pattern $P_\beta$ can be decomposed in a Flip (or commutator) widget $F_\beta$, three corner widgets $C_\beta$ and link widgets $L_\beta$.
  
  The corner widget $C_\beta$ and link widget $L_\beta$ both appear in two variants $C_\beta^0$, $C_\beta^1$ and $L_\beta^0$, $L_\beta^1$.
  The flip widget $F_\beta$ appears in four variants $F_\beta^0$, $F_\beta^1$, $F_\beta^2$ and $F_\beta^3$.

  So the ant goes around the rectangle four times before it reaches a configuration that is identical up to having more links on all four sides.
  We call those four go-around a \emph{full-revolution}.
  Starting from configuration $P_\beta$, the $k$th full-revolution takes $600+128k$ steps, and at each full-revolution the ant moves by $(2,2)$, and each side of the rectangle has 2 new $L_\beta$ widgets.

  Starting on the pattern $P_\gamma(k)$ of Fig.~\ref{fig:pattern_diagonal} composed of a main widget $M_\gamma$, $k$ linking widgets $L_\gamma$ and a bouncer widget $B_\gamma$, the ant $LLRRRL$ starts a diagonal highway\footnote{\href{https://lutfalla.fr/ant/highway.html?antword=LLRRRL&t=diagonal&n=3}{\texttt{lutfalla.fr/ant/highway.html?antword=LLRRRL\&t=diagonal\&n=3}}} of period $800 + 96k$ and drift $(4,4)$.

  \begin{figure}[htp]
    \center
    \includegraphics[width=0.4\textwidth]{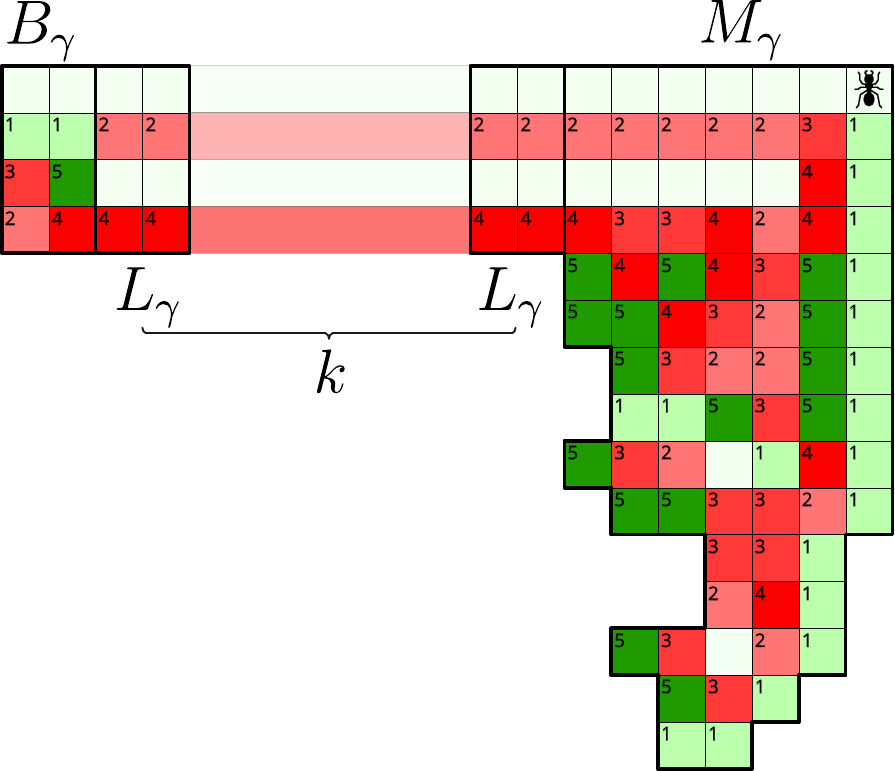}
    \caption{The pattern $P_\gamma(k)$ consisting of a $B_\gamma$ widget, $k$ aligned $L_\gamma$ widgets and a $M_\gamma$ widget on which the ant $LLRRRL$ starts a highway of period $800+96k$ and drift $(2,2)$.}
    \label{fig:pattern_diagonal}
  \end{figure}
  The pattern $P_\gamma(k)$ is quite complex as the bouncer widget $B\gamma$ has only $2$ variants $B_\gamma^0$ and $B_\gamma^1$, whereas the link widget $L_\gamma$ has 6 variants $L_\gamma^0$ to $L_\gamma^6$, and the main widget $M_\gamma$ has $6$ variants $M_\gamma^0$ to $M_\gamma^6$, and most times the ant only traverses a small portion of the $M_\gamma$ widget.
  So the ant bounces 6 times between the $B_\gamma$ and $M_\gamma$ widgets in a full highway period.

\end{proof}

\section{Emergent behaviours of the $LLRLRLL$ ant}
\label{sec:llrlrll}
In the previous section we discussed ant $LLRRRL$ which, in particular, has a non-highway emergent behaviour : a regular increasing-rectangle behaviour.

In this section we present ant $LLRLRLL$ which has both an infinity of highway behaviours, and regular cone behaviours.

\begin{theorem}
  The $LLRLRLL$ ant admits cone behaviours and an infinite family of diagonal highways.
  %% For any $k$, the $LLRLRLL$ ant admits a highway of period $208+24k$ and drift $(\pm 2, \pm 2)$.\\
\end{theorem}

\begin{proof}[Sketch]
  This result is proved using widgets.
  We provide in this sketch of proof the widget and a short description of their dynamics. %% FIXME appendix

  Starting on the pattern $P_\zeta(k)$ shown in Fig.~\ref{fig:llrlrll_diagonal} the ant $LLRLRLL$ starts a diagonal highway\footnote{\href{https://lutfalla.fr/ant/highway.html?antword=LLRLRLL&t=highway&n=3}{\texttt{lutfalla.fr/ant/highway.html?antword=LLRLRLL\&t=highway\&n=3}}} of period $208+24k$.
  The pattern $P_\zeta$ is composed of a main widget $M_\zeta$, $k$ link widgets $L_\zeta$ and a bouncer widget $B_\zeta$. The main widget $M_\zeta$ and bouncer widget $B_\zeta$ have two versions, whereas the link widget has four. Hence a full highway period corresponds to 2 back-and-forth.

  \begin{figure}[htp]
    
    \center
    \includegraphics[width=0.4\textwidth]{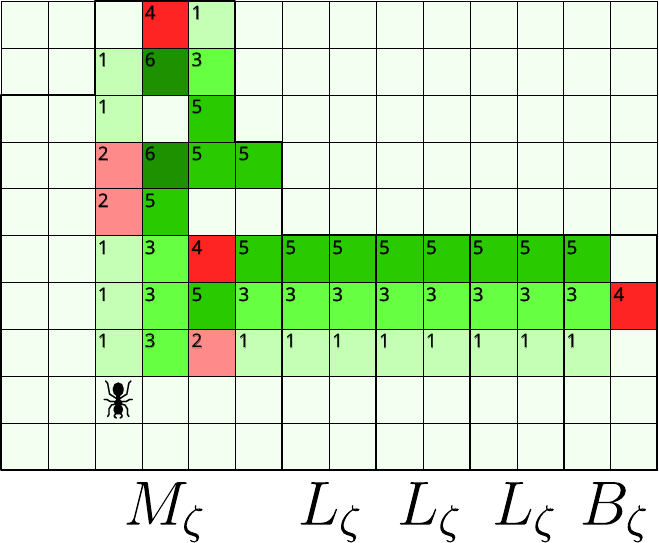}
    \caption{The pattern $P_\zeta(3)$ starting a diagonal highway of period $280$ for ant $LLRLRLL$.} %% TODO more than 232
    \label{fig:llrlrll_diagonal}
\end{figure}

  Starting on the pattern $P_\theta$ of Fig.~\ref{fig:cone}, the ant $LLRLRLL$ starts a cone behaviour\footnote{\href{https://lutfalla.fr/ant/highway.html?antword=LLRLRLL&t=cone}{\texttt{lutfalla.fr/ant/highway.html?antword=LLRLRLL\&t=cone}}}.
The pattern $P_\theta$ is composed of a main widget $M_\theta$ and a turn-and-bounce widget $B_\theta$.
The main widget has $3$ versions and the turn-and-bounce widget has 2. Hence a full repetition corresponds to 6 back-and-forth of the ant.
The first full repetition takes $780$ steps, at the end of which the configuration now contains the main widget $M_\theta$, two link widgets $L_\theta$ and the turn-and-bonce widget $B_\theta$ (Fig.~\ref{fig:cone} centre).
The $k$-th full repetition takes $780 + 192k$ steps, adds two link widgets and moves the ant by $(-6,-2)$.

\begin{figure}[htp]
  %% TODO
  \center
  \includegraphics[width=\textwidth]{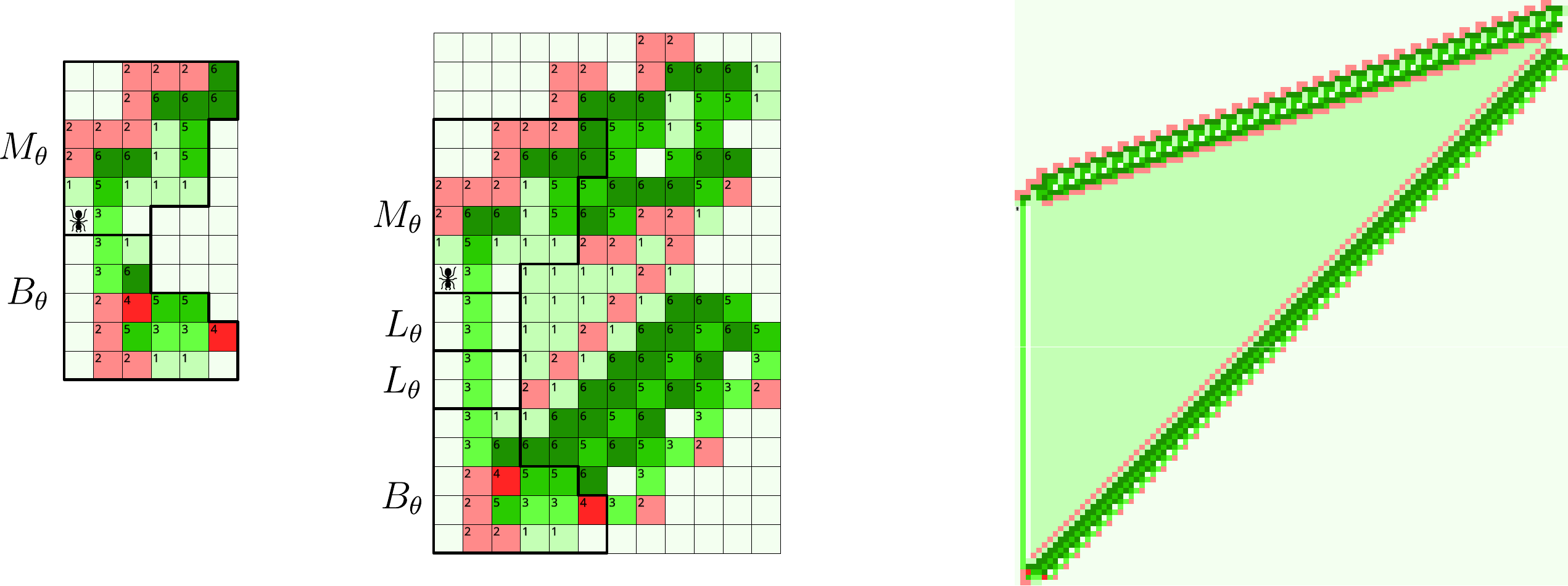}
  \caption{The pattern $P_\theta$ that starts a cone (left), and the configurations reached from $P_\theta$ after $780$ (centre) and $35520$ steps (right, scaled down).}
    \label{fig:cone}
\end{figure}

\end{proof}

\section{Further work : recurrent cells}

In \cite{gajardo2025} and in the present paper we argue that no generalisation of the highway conjecture holds for generalised ants as there exist ants which have infinitely many highways and ants where both highway and non-highway behaviours can be reached from initially finite configurations.

However, our research and computer explorations point us towards a new conjecture concerning all initial configurations (and not only finite configurations) : 

\begin{conjecture}
  For any non-trivial ant $w$ and any initial configuration $C$,
  the set $\mathrm{Rec}_{w,C}\subseteq\mathbb{Z}^2$ of positions that are visited infinitely many times is either $\emptyset$ or $\mathbb{Z}^2$.
\end{conjecture}

%% \bibliography{biblio.bib} %% only for llncs
\printbibliography
\end{document}